\documentclass[11pt]{article}

\usepackage{microtype}
\usepackage[utf8]{inputenc}
\usepackage{lmodern}
\usepackage[T1]{fontenc}
\usepackage{hyperref}

\usepackage{aliascnt}
\usepackage{amssymb}
\usepackage[tbtags,fleqn]{amsmath}
\usepackage{amsfonts}
\usepackage{amsthm}
\usepackage{array}
\usepackage[USenglish]{babel}
\usepackage[capitalise,noabbrev]{cleveref}
\usepackage{color}
\usepackage{comment}
\usepackage{enumerate}
\usepackage[figuresright]{rotating}
\usepackage{graphicx}
\usepackage[labelsep=space,singlelinecheck=false,font={up,small},labelfont={sf,bf},list=false]{caption}
\usepackage{lastpage}
\usepackage[left,mathlines]{lineno}
\usepackage{listings}
\usepackage[mathscr]{eucal}
\usepackage{multirow}
\usepackage[online]{threeparttable}
\usepackage{soul}
\usepackage{subcaption}
\usepackage{tabularx}
\usepackage{textcomp}
\usepackage{xstring}

\usepackage{amsfonts}
\usepackage{float}
\usepackage[ruled,vlined,commentsnumbered]{algorithm2e}
\usepackage{algorithmic}
\usepackage[title]{appendix}
\usepackage{todonotes}

\usepackage{mathtools}
\usepackage{tikz}
\usetikzlibrary{arrows,shapes.geometric,shapes.misc}

\usepackage{csquotes}
\usepackage[numbers,sort]{natbib}

\graphicspath{{./imgs/}}

\crefformat{page}{#2page~#1#3}%
\Crefformat{page}{#2Page~#1#3}%
\crefformat{equation}{#2(#1)#3}%
\Crefformat{equation}{#2(#1)#3}%
\crefformat{figure}{#2Figure~#1#3}%
\Crefformat{figure}{#2Figure~#1#3}%
\crefformat{section}{#2Section~#1#3}
\Crefformat{section}{#2Section~#1#3}
\crefformat{appendix}{#2Appendix~#1#3}
\Crefformat{appendix}{#2Appendix~#1#3}
\crefformat{chapter}{#2Chapter~#1#3}
\Crefformat{chapter}{#2Chapter~#1#3}
\crefformat{chapter*}{#2Chapter~#1#3}
\Crefformat{chapter*}{#2Chapter~#1#3}
\crefformat{part}{#2Part~#1#3}
\Crefformat{part}{#2Part~#1#3}
\crefformat{enumi}{#2(#1)#3}
\Crefformat{enumi}{#2(#1)#3}
\crefname{algocf}{Algorithm}{Algorithms}
\Crefname{algocf}{Algorithm}{Algorithms}

\newtheorem{theorem}{Theorem}
\newtheorem{corollary}[theorem]{Corollary}
\newtheorem{definition}[theorem]{Definition}
\newtheorem{lemma}[theorem]{Lemma}

\newcommand*{\defeq}{\mathrel{\vcenter{\baselineskip0.5ex \lineskiplimit0pt
                     \hbox{\scriptsize.}\hbox{\scriptsize.}}}%
                     =}

\makeatletter
\def\Ddots{\mathinner{\mkern1mu\raise\p@
\vbox{\kern7\p@\hbox{.}}\mkern2mu
\raise4\p@\hbox{.}\mkern2mu\raise7\p@\hbox{.}\mkern1mu}}
\makeatother

\newcommand{\Mopt}{M}
\newcommand{\Malg}{D}
\newcommand{\outtree}{\mathcal{T}}
\newcommand{\allouttrees}{\mathcal{T}}
\newcommand{\local}{LOCAL\xspace}
\newcommand{\congest}{CONGEST\xspace}
\newcommand{\mingirth}{\ensuremath{4r+3}}
\newcommand{\loweralg}{\ensuremath{\mathcal{ALG}}\xspace}

\newcommand{\N}{\mathbb{N}}

\newcommand{\sizeof}[1]{\left|#1\right|}

\DeclareMathOperator*{\argmax}{arg\,max}
\DeclareMathOperator*{\argmin}{arg\,min}
\newcommand{\varprio}[1]{\mathord{\mathit{prio}}^{#1}}
\newcommand{\varid}[1]{\mathord{\mathit{sel}}^{#1}}
\newcommand{\distr}{distance-$r$ dominating set\xspace}

\newcommand{\DISTR}{Distance-$r$ Dominating Set\xspace}

\newenvironment{apthm}[1]{\par\addvspace{3mm}\noindent\textbf {\Cref{#1}\;}}{\par\addvspace{3mm}}

\title{Distributed \DISTR \\ on Sparse High-Girth Graphs}

\author{
  Saeed Akhoondian Amiri\thanks{Max Planck Institute for Informatics, Saarland Informatics Campus, Germany; \texttt{\{samiri,bwiederh\}@mpi-inf.mpg.de}}
  \and
  Ben Wiederhake$^{\ast}$\thanks{Saarbrücken Graduate School of Computer Science, Germany}
}

\begin{document}

\maketitle

\begin{abstract}
\noindent The dominating set problem and its generalization, the
\distr problem, are among the well-studied problems in the sequential
settings. In distributed models of computation, unlike for domination,
not much is known about distance-$r$
domination. This is actually the case for other important
closely-related covering problem, namely, the distance-$r$ independent set problem.

By result of Kuhn et
al.~\cite{Kuhn:2016:LCL:2906142.2742012} we know the distributed
domination problem is hard on
high girth graphs; we study the
problem on a slightly restricted subclass of these graphs:
graphs of bounded expansion with high girth, i.e.\ their girth should be
at least $\mingirth$.

We show that in such graphs, for every constant $r$, a simple greedy \congest
algorithm provides a constant-factor approximation of the minimum
distance-$r$ dominating set problem, in a
constant number of rounds. More precisely, our constants are dependent
to $r$, not to the size of the graph. This is the first algorithm that
shows there are non-trivial constant
factor approximations in constant number of rounds for any distance
$r$-covering problem in distributed settings. 

To show the dependency
on $r$ is inevitable, we provide an unconditional lower bound showing the same
problem is hard already on rings. We also show that our analysis of
the algorithm is relatively tight, that is any significant improvement to the
approximation factor requires new algorithmic ideas.
\end{abstract}

\clearpage

\section{Introduction}\label{sec:intro}

The dominating set problem asks for a set of vertices $D$ of a
graph $G=(V,E)$ such that every other vertex of $G$ is a neighbor of a
vertex in $D$. Given that $D=V$ is a trivial solution to the problem,
we are interested in finding a set $D$ of small size.

The problem plays an important role both
from the theoretical and practical perspective of computer science. For the latter,
it serves as an initial set of vertices to form a network backbone
and e.g.~facilitates constructing small routing tables in networks; for
the former, the dominating set problem is a central problem in
showing lower bounds for several complexity paradigms: It is one of
Karp's 21 NP-complete problems, it is the
central $W[2]$-complete.

We consider the problem in distributed settings, namely on \local and
\congest models. Intuitively speaking, in these models, every
vertex in the graph is a processor, has a unique identifier,
and communicate only with its neighbors per round.
The \congest model restricts the bandwidth of communication links to a
reasonable complexity. The aim is to solve the problem with the least
number of communication rounds.
A more rigorous definition follows in~\Cref{sec:model}.
We specifically look into the problem of finding a small \distr,
where each vertex needs to output its membership.

In the following, we first go over the status of domination problems in distributed computing and in particular in sparse graphs. Then we
explore existing tools and we explain what is their shortage for our
purpose. Afterward, we introduce our results and
determine where in the existing literature it belongs to.

\subsection*{Related Work}
% In distributed settings, Kuhn et al.~\cite{Kuhn:2016:LCL:2906142.2742012} showed
% and in distributed settings, and it is one of the well-known problems without
% logarithmic approximation in a constant number of rounds; a negative
% example of the famous question of Naor and Stockmeyer: \enquote{What can be
% computed locally?}~\cite{what-local}

In distributed settings, for dominating set problem in general graphs recently
Kuhn et al.~\cite{DBLP:conf/podc/DeurerKM19} provide a $ (1 +\epsilon)(1 +\log (\Delta+
1))$-approximation of the problem in $f(n)$ rounds, where $\Delta$ is
the maximum degree and
$f\colon\N\rightarrow\N$ is the number of rounds that is needed to
compute a special graph decomposition, called the network
decomposition~\cite{DBLP:conf/focs/AwerbuchGLP89,Awerbuch:1992:FND:135419.135456,DBLP:conf/soda/Ghaffari19}.

On the other hand, the lower bound of Kuhn et
al.~\cite{Kuhn:2016:LCL:2906142.2742012} shows that finding a
logarithmic approximation in sublogarithmic time for
minimum dominating set (and some other covering problems) is impossible
in general graphs, even in the \local model of computation. Their
lower-bound graph has a high girth (as a function of $n$), but also it was
of unbounded arboricity (more generally unbounded average degree).

In fact, they provide a negative example of the famous question of Naor and Stockmeyer: \enquote{What can be
computed locally?}~\cite{what-local}. 

If we consider a graph class of very high girth and very low edge
density, e.g.\ trees (graphs of infinite girth), the problem is easy to approximate
in zero rounds: take all non-leaf vertices.

The above observations raise the following question: 
In which graph classes does the problem admit a constant
approximation factor in a constant number of rounds? 

Naturally, given the lower bound, we have to search affirmative answer
in sparse graphs.
Along this line, there are several interesting results in
sparse graphs. Lenzen et
al.~\cite{10.1007/978-3-540-87779-0_6,ds-planar} provided the first
constant-factor approximation in a constant number of rounds in planar
graphs, then this has been improved by Czygrinow et al.~\cite{10.1007/978-3-540-87779-0_6}. Later Amiri
et al.~\cite{Amiri2016,amirilog} provided a new analysis method to
extend the result of Lenzen et al.\ to bounded genus graphs. This
has recently been improved to excluded minor graphs by
Czygrinow et al.~\cite{DBLP:conf/isaac/CzygrinowHWW18}.

A natural generalization of excluded minor graphs is the class of
bounded expansion graphs, in simple words, bounded expansion graphs
also exclude minors but only locally; we may have large clique minor
in the entire graph. 

On graphs of bounded expansion, there is only a logarithmic time constant
factor approximation known for dominating set, however, it seems that
one can extend the algorithm of ~\cite{DBLP:conf/isaac/CzygrinowHWW18}
to bounded expansion graphs as they care only about local minors. 
If we go slightly beyond those graphs, to graphs of
bounded arboricity (where every subgraph has a constant edge density), the situation is worse: only an
$O(\log \Delta)$-approximation in $O(\log n)$ rounds is known. There is a
$O(\log n)$ round $O(1)$-approximation in such graphs, however, this algorithm is
randomized~\cite{ds-arbor}.

What said is all about dominating set, the situation gets drastically
worse if we go to \distr. That is the only known algorithm that solves
the problem on a non-trivial class of graphs (e.g.~paths, cliques are
among trivial classes of graphs) is the algorithm of Amiri et
al.~\cite{amiri2017distributed} for bounded expansion graphs which
provides a constant factor approximation in a logarithmic number of
rounds.

\subsection*{The Challenge of Approximating \DISTR}

There are several existing approaches one might try to employ to
tackle the problem: $1)$ take the $r$-th power of the graph and go back to dominating set, $2)$
decentralized existing decomposition methods in the sequential setting
and employ them, $3)$ use existing fast distributed graph
decomposition methods for sparse graphs. In the following, we explain how all of the
above approaches are not practical in providing sublogarithmic
round algorithms for distance-$r$ covering problems, and
in particular for \distr.

For the first approach, clearly we lose the sparsity of the graph
already on stars. Hence, we cannot rely on existing algorithms for domination
problem in sparse graphs. 

If we decentralize the existing sequential
decomposition methods, we can barely hope to get anything better
than logarithmic rounds: every such decomposition we know % no comma here!
already takes
polylogarithmically many rounds. Even assuming the decomposition is
already given, in such methods we have to sequentially go over the
clusters, however the number of clusters is usually at least
logarithmic. Hence, we cannot
hope for a sublogarithmic rounds.

For the third approach, we first briefly explain how the existing
methods work and then explain why it is not possible to stick to known
techniques. These methods are mostly inspired by existing methods in
classical settings, like Baker's
method~\cite{baker1994approximation} e.g.\ the $O(\log^\ast n)$ round
algorithm of~\cite{czygrinow2006distributed} is among them. The idea
is to find a partition of
a sparse graph into connected clusters such that each cluster has a
small diameter and the number of in-between cluster edges is small (in the Baker's
method, we have treewidth of each cluster is bounded, whereas for
distributed settings we require the diameter to be small). Then find
the optimal solution inside each cluster efficiently, and given the
fact that the number of direct edges between a pair of
clusters is small, we ignore/resolve conflicts. 

However, the latter already for distance-$2$
domination fails, as the number of in-between cluster edges in the
distance graph is high. Also, we do not have a pleasure to rely on global properties
similar to what we exploit in the sequential setting to make our
choice wiser. Since such approaches increase the number of
rounds. Therefore, any distributed algorithm that solves distance covering
problems either have to develop completely new techniques or resolve the above
issue concretely tailored for the underlying graph class/problem.

\subsection*{Our Results}
As said, we consider a generalization of the dominating set problem,
i.e.\ the~\distr problem: find a set of vertices
$D$ such that every other vertex is within distance $r$ of one of the
vertices in $D$. We fill a gap between the lower bound and upper
bounds by analyzing the complexity of the problem on
graphs of high girth (similar to the lower bound
graph by Kuhn et al.), but given that lower bound graph was relatively
dense, we restrict
ourselves to sparse graphs, in particular to bounded expansion graphs
(similar to the work of~\cite{amiri2017distributed}). 

In the aforementioned class of graphs, we
provide an affirmative answer to the question of Naor and Stockmeyer, by
designing a simple deterministic \congest constant factor
approximation algorithm (more precisely, our constant depends only on
$r$) in a constant number of rounds (again, this constant depends only
on $r$ not the size of the graph). Formally, we prove the
following~\Cref{thm:algorithm} in~\Cref{thm:mdsapprox}.

\begin{theorem}\label{thm:algorithm}
    % Also change the copy!
    Let $\mathcal{C}$ be a graph class of bounded expansion $f(r)$ and girth at
    least $\mingirth$. There is a \congest algorithm that runs in $O(r)$ rounds
    and provides an $O(r\cdot f(r))$-approximation of minimum
    \distr on $\mathcal{C}$.
\end{theorem}

As explained earlier, in contrast to our algorithm, the existing
methods for domination problems are not extendable to distance-$r$
domination problem.

% More generally we are not aware of any constant factor approximation in
% a constant number of rounds (or more generally sublogarithmic number of
% rounds) in any non-trivial graph class (e.g.~not restricted to bounded
% degrees) for the \distr problem.

Given that distance-$r$ dominating set is equivalent to the dominating
set of the $r$-th power of graph, this is one of the few algorithms
that can actually provide a constant factor approximation in
a non-trivial class of dense graphs for covering problems. 
There are very few known algorithms
with a constant factor guarantee in a constant number of rounds on
non-trivial dense graphs, e.g.~the algorithm of Schneider et
al.~\cite{DBLP:conf/podc/SchneiderW08} on graphs of bounded
independence number (for the independent set and the connected
dominating set problem), partially falls in this category.

To show that our upper bound is reasonably tight, we provide a lower bound as well.
This we obtain by a reduction from lower bound for independent set on the ring
~\cite{10.1007/978-3-540-87779-0_6,ds-planar} to the \distr on rings
(of course, the girth of the ring is high). 
More formally we prove the
following~\Cref{thm:lowerbound}.
\begin{theorem}\label{thm:lowerbound}
    % Also change the copy!
    Assume an arbitrary but fixed $\delta > 0$ and $r > 1$, with $r \in o(\log^{\ast} n)$.
    Then there is no deterministic \local algorithm that finds in $O(r)$ rounds
    a $(2r + 1 - \delta)$-approximation of \distr for all $G \in \mathcal{C}$,
    where $\mathcal{C}$ is the class of cycles of length $\gg \mingirth$.
\end{theorem}

We will formally introduce the notion of bounded expansion in the
preliminaries, for the moment it is fine to imagine these as
generalizations of bounded degree graphs and $H$-minor-free
graphs. For more information on the relation between sparse graph
classes we refer the reader to~\Cref{fig:classes}.

\begin{figure}[t]
  \centering
  \includegraphics[scale=0.38]{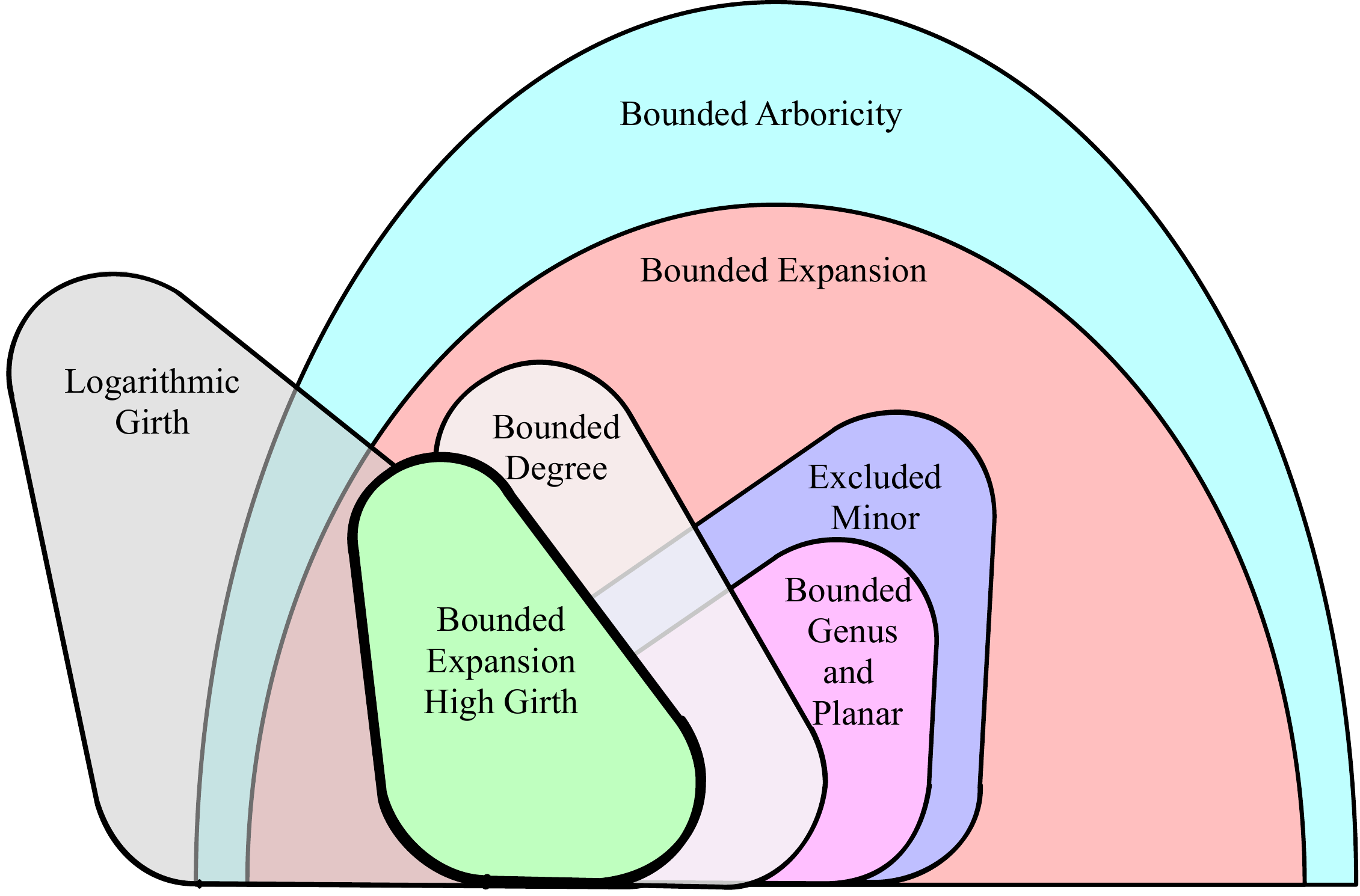}
  \caption{Diagram of the relation of sparse graph classes. The graph class
  in lower bound construction of Kuhn et al.~\cite{Kuhn:2016:LCL:2906142.2742012}
  is a subclass of logarithmic girth class.
  The bounded expansion class is a
  subclass of bounded arboricity class.  Bounded expansion is also a superclass
  of many common sparse graph classes: planar, bounded genus, excluded
  minors, and bounded
  degree. The class of bounded expansion with high girth intersects
  each of the other four classes, but neither contains nor is fully contained in any of them.}
  \label{fig:classes}
\end{figure}

\subparagraph*{Paper organization:}
First, we introduce our notation in~\Cref{sec:model},
then in~\Cref{sec:algo} we explain the algorithm, and
an analysis of its running
time, correctness, and approximation guarantee is given.
We also show that our analysis is asymptotically tight by providing an example
that matches the analysis. In~\Cref{sec:hardness} we provide a
lower bound construction and its correctness proof.
Finally, we conclude our work in~\Cref{sec:conclude}.

\section{Preliminaries}\label{sec:model}

We assume basic familiarity with graph theory.
In the following, we introduce
some of the basic graph notations to avoid ambiguities. We refer the reader
to the book by Diestel~\cite{DBLP:books/daglib/30488} for further reading.

\subparagraph*{Graph, Neighborhood, Distance-$r$:}
We will only consider simple, undirected graphs $G=(V,E)$.
For $u, v \in V$, define
$d(u,v)$ as the distance (in number of edges) between the two vertices.
For a set $S \subseteq V$, we define $d(u,S)$ as the
distance between vertex $u$ and any vertex in $S$.

Two vertices
$u,v\in V$ are neighbors in $G$ if there is an edge
$e=\{u,v\}\in E$.
We extend this definition to the distance-$r$ neighborhood $N^r[v]$ and open distance-$r$
neighborhood $N^r(v)$ of a vertex $v$ in the following way:
\begin{align*}
  N^{r}[v] &\defeq \{ u \in V \mid d(u,v) \leq r \} \\
  N^{r}(v) &\defeq N^r[v] \setminus \{v\}
\end{align*}
Similarly for a set $S \subseteq V$ we define: 
\begin{align*}
  N^r[S] &\defeq \bigcup_{v\in S} N^r[v] \\
  N^r(S) &\defeq N^r[S] \setminus S
\end{align*}

\subparagraph*{Girth, Radius:}
The girth $g$ of a graph $G$ is the length of its shortest cycle, or $\infty$ if acyclic.
The radius $R$ of $G$ is the minimum integer $R$ such that there is a
vertex $v\in V$ so that $N^R[v]=V$.

\subparagraph*{\DISTR:}
A set $M\subseteq V$ is a
\distr if $V=N^r[M]$. If additionally there is no smaller such set,
then $M$ is a minimum \distr
of $G$.

\subparagraph*{Edge Density, $r$-Shallow Minor, Expansion:}
Let $G=(V,E)$ be a graph, its edge density is $|E|/|V|$. A graph $H$
is an $r$-shallow minor of $G$ if $H$ can be obtained from $G$ by the
following operations. 
First, we take subgraph $S$ of $G$ and then partition vertices of $S$ into vertex disjoint
connected subgraphs $S_1,\ldots,S_t$ of $S$, each of them of radius at most $r$ and, at
the end contract each $S_i$ ($i\in [t]$) to a single vertex to obtain
$H$. We denote by $\nabla_r(G)$
the maximum edge density among all $r$-shallow minors of the graph $G$.

A graph class $\mathcal{C}$ is bounded expansion if there is a
function $f:\N\rightarrow\N$ such that for every graph
$G\in\mathcal{C}$ and integer $r\in \N$ we have $\nabla_r(G)\le
f(r)$; here $f$ is the \emph{expansion function}. A class of graphs
$\mathcal{C}$ has constant expansion if for every integer $r$ we have
$f(r)\in O(1)$. 

Every planar, bounded genus and, excluded minor graph is a
constant expansion graph. Every class of bounded degree graphs is also
bounded expansion, but not of constant expansion. 
For more information on bounded expansion
graphs, we refer the reader to the book of Ne{\v{s}}et{\v{r}}il and
Ossona de Mendez~\cite{nevsetvril2012sparsity}.

\subparagraph*{\local and \congest model of computation:}
The \local model of
computation assumes that the problem is being solved in a distributed manner:
Each vertex in the graph is also a computational node, the input graph
is also the communication graph, and initially, each vertex only knows its own unique identifier and its neighbors.
An algorithm proceeds in synchronous rounds on each vertex in parallel.
In each round, the algorithm can run an arbitrary amount of local computation,
send a message of arbitrary size to each of the neighboring vertices, and then receive all messages from its neighbors.
At the end of the round, each vertex can decide locally whether it wants to halt with an output, or continue.
The most common metric is the number of synchronous rounds, i.e.~the number of communication rounds.

This model first introduced by Linial~\cite{linial}, later
Peleg~\cite{Peleg:2000:DCL:355459} named it \local model.

The \congest model is very similar to the \local model, except
that identifiers can be represented in $O(\log n)$ bits,
and each message can only hold $O(\log n)$ bits,
where $n$ is the number of vertices in the network.

\section{Distributed Approximation Algorithm for Dominating Set}\label{sec:algo}

We present an algorithm that solves $r$-MDS with approximation factor
$O(r \cdot f(r))$ in time $O(r)$ on graphs with girth at least $\mingirth$.
Specifically, we provide an algorithm that proves the following theorem.

\begin{apthm}{thm:algorithm}
    Let $\mathcal{C}$ be a graph class of bounded expansion $f(r)$ and girth at
    least $\mingirth$. There is a \congest algorithm that runs in $O(r)$ rounds
    and provides an $O(r\cdot f(r))$-approximation of minimum
    \distr on $\mathcal{C}$.
\end{apthm}

We prove this by providing \Cref{alg:rmds} satisfying all bounds.
At its core, the algorithm is simple:
Each vertex computes the size of its distance-$r$ neighborhood, i.e.~the distance-$r$ degree.
This degree is propagated, so that each vertex selects in its distance-$r$ neighborhood the vertex with the highest such degree.
The output is the set of all selected vertices.
We expect this to yield a good approximation because only few candidates can be selected.

\Cref{alg:rmds} defines this formally.
The main technical contribution is \Cref{thm:mdsapprox},
which concludes that \Cref{alg:rmds} is correct and satisfies all bounds in \Cref{thm:algorithm}.

\begin{algorithm}[h]
\caption{CONGEST computation of $r$-MDS, on each vertex $v$ in parallel}
\begin{algorithmic}[1]
\STATE Compute $\sizeof{N^{r}(v)}$, e.g.~using \Cref{alg:neighborhood}
\STATE // Select the vertex with the highest degree:
\STATE $(\varprio{v}, \varid{v}) \defeq (\sizeof{N^{r}(v)}, v)$
\FOR{$r$ rounds}
    \STATE Send $(\varprio{v}, \varid{v})$ to all neighbors
    \STATE Receive $(\varprio{u}, \varid{u})$ from each neighbor $u$
    \STATE $(\varprio{v}, \varid{v}) \defeq \max_{u \in N^{1}[v]}\{ (\varprio{u}, \varid{u}) \}$ \label{algline:maxtuples}
    \STATE Remember all received messages that contained $(\varprio{v}, \varid{v})$
\ENDFOR{} \label{algline:votingcomplete}
\STATE // Propagate back to the selected vertex:
\STATE $\Malg^{v} \defeq \{ \varid{v} \}$
\FOR{$r - 1$ rounds}
    \FOR{each neighbor $u \in N^{1}(v)$}
        \STATE Determine which vertices sent by $u$ are in $\Malg^{v}$
        \STATE Send these to $u$, encoded as a bitset of size $r$
    \ENDFOR{}
    \STATE Receive bitsets, extend $\Malg^{v}$ accordingly
\ENDFOR{}
\RETURN $v \in \Malg^{v}$
\end{algorithmic}\label{alg:rmds}
\end{algorithm}

We say that~\Cref{alg:rmds} computes a set $\Malg$ by returning $\top$
for all vertices in the set, and $\bot$ for all others.
Naturally, messages and comparisons only consider the ID of vertices, and not the vertices themselves.
This abuse of notation simplifies the algorithm and analysis.
In line \ref{algline:maxtuples} we order tuples lexicographically:
Tuples are ordered by the first element (the size of the distance-$r$ neighborhood); ties are broken by the second element (the ID of the vertex).

\begin{algorithm}[h]
\caption{\congest computation of $\sizeof{N^{r}(v)}$, on each vertex $v$ in parallel}
\begin{algorithmic}[1]
\STATE $n_{u} \defeq 1$ for all $u \in N^{1}(v)$
\FOR{$r - 1$ rounds}
    \STATE To each vertex $u \in N^{1}(v)$, send $1 + \sum_{w \in N^{1}(v) \setminus \{u\}} n_{w}$ \label{algline:computetree}
    \STATE $n_{u} \defeq$ the number received from $u$, for each $u \in N^{1}(v)$
\ENDFOR{}
\RETURN $\sum_{w \in N^{1}(v)} n_{w}$ \label{algline:computeneighborhood}
\end{algorithmic}\label{alg:neighborhood}
\end{algorithm}

Observe that counting the neighborhood of the
root in a rooted tree is easy; that the graph looks like a tree locally;
and that the messages sent across an edge (if any) are identical no matter which vertex is the root.
Line \ref{algline:computetree} avoids informing a vertex about its own subtree.

The remainder of this section proves the correctness, running time, and approximation factor for the algorithm.

\subsection*{Correctness}

First we will show basic correctness properties.
One can trivially check that all messages contain only $O(\log n)$ many bits.
Specifically, the bitsets have only size $r \in o(\log^{\ast} n)$.

\begin{lemma}
  \Cref{alg:neighborhood} computes the size of $N^{r}(v)$ for each vertex $v$ in parallel.
\end{lemma}
\begin{proof}
  First, observe that in only $r - 1$ rounds of communication, no cycle can be detected, as the girth is at least $\mingirth$.
  This means that $N^{i}(v)$ is a tree for every $i \leq r-1$ and $v \in V$.
  We define the tree $T_{u, i}^{\neg v}$ as the (set of vertices in the) tree of edge-depth $i$, rooted at vertex $u$,
  excluding the branch to vertex $v$, where $v$ is a neighbor of vertex $u$.

  Now we can prove by induction: At vertex $v$, after the $i$-th round\footnote{
  We interprete \enquote{after the zeroth round} as \enquote{before the first round}}{}
  (where $0 \leq i \leq r - 1$),
  $n_{u}$ stores the size of the tree $T_{u, i}^{\neg v}$.

  For the induction basis $i=0$, we know $\forall u, v: n_{u} = 1 = \sizeof{T_{u,0}^{\neg v}} = \sizeof{\{u\}}$.

  This leaves the induction step: At the beginning of the $i$-th round (for $1 \leq i \leq r-1$),
  we know that $n_{u} = \sizeof{T_{u, i - 1}^{\neg v}}$ by the induction hypothesis, for every $u, v$.
  Consider vertex $v$.  By construction, its distance-$i$ open neighborhood is the
  union of all edge-depth $i-1$ trees of $v$'s neighbors, so: $N^{i}(v) = \bigcup_{u \in N^{1}(v)} T_{u,i - 1}^{\neg v}$.
  Due to the high girth requirement, we know that all sets in this union are disjoint.
  Vertex $v$ can therefore compute $\sizeof{N^{i}(v)}$ by summing up all its $n_{u}$s,
  and can even compute $\sizeof{T_{v,i}^{\neg u}}$ for an arbitrary vertex $u$ by subtracting the corresponding $n_{u}$ again.
  This is exactly what happens in line \ref{algline:computetree}.
  Then $v$ sends $\sizeof{T_{v,i}^{\neg u}}$ to each neighbor $u$, which stores it in the corresponding variable $n_{v}$.
  By symmetry, this also means that vertex $v$ now has stored $\sizeof{T_{u,i}^{\neg v}}$ in $n_{u}$, thus proving the induction step.

  With the meaning of $n_{u}$ established, line~\ref{algline:computeneighborhood}\ must compute $\sizeof{N^{r}(v)}$.
\end{proof}

Next we show that \Cref{alg:rmds} selects the maximum degree neighbor:

\begin{lemma}
  \label{lem:selectneighbor}
  In \Cref{alg:rmds}, when the selection phase is over (line \ref{algline:votingcomplete} and onward),
  each vertex $v$ has selected a vertex $\varid{v}$.
  This is the unique vertex $\argmax_{u \in N^{r}[v]} \{ ( \sizeof{N^{r}(u)}, u ) \}$.
\end{lemma}
\begin{proof}
  By construction, only tuples of the form $( \sizeof{N^{r}(w)}, w )$ with $w \in V$ are ever stored.
  The $\max$ operator is commutative and associative, so it is sufficient to prove that
  each vertex $v$ considers precisely the tuples for $w \in N^{r}[v]$.
  This can be shown by straightforward induction: After round $i$, vertex $v$ considers precisely the tuples for $w \in N^{i}[v]$.
  The base case is $i = 0$, the induction step is straight-forward.
\end{proof}

This shows that each vertex $v$ selects the maximum vertex in $v$'s neighborhood.
Next, we show that this selection is back-propagated:

\begin{lemma}
  If there is a vertex $u$ that selects $v$ ($\varid{u} = v$), then $v \in \Malg^{v}$.
\end{lemma}
\begin{proof}
  Consider the path along which $v$ was forwarded % Want to avoid the term "propagate" here
  during the selection phase.
  By straight-forward induction one can see that after $i$ rounds of propagation,
  the first $i$ many vertices $w$ on this path (starting at $u$) have $v \in \Malg^{w}$.
  The path has length at most $r$ edges, so $v \in D^{v}$ after $r$ rounds.
\end{proof}

And because no further vertices are added into any $\Malg^{v}$, we get:

\begin{corollary}
  \label{cor:selectediscomputedset}
  The selected vertices are precisely the computed set:
  \begin{equation}
    \Malg = \{ \varid{v} \mid v \in V \}
  \end{equation}
\end{corollary}

Together with \Cref{lem:selectneighbor}, this shows that \Cref{alg:rmds} indeed computes a dominating set:

\begin{lemma}
  The computed set $\Malg$ is a distance-$r$ dominating set.
\end{lemma}
\begin{proof}
  Assume towards contradiction that a vertex $v$ is not dominated.
  \Cref{lem:selectneighbor} shows that $v$ selected a vertex $\varid{v}$ in its distance-$r$ neighborhood.
  \Cref{cor:selectediscomputedset} shows that $\varid{v} \in \Malg$,
  which distance-$r$ dominates $v$ in contradiction to the original assumption.
\end{proof}

The time complexity analysis is trivial:

\begin{lemma}
  \Cref{alg:rmds} runs in $O(r)$ rounds.
\end{lemma}
\begin{proof}
  Each loop takes $O(r)$ iterations, and each iteration takes constantly
  many rounds. Therefore, the overall algorithm takes $O(r)$ rounds.
\end{proof}

This concludes the basic correctness properties.
What remains to be shown is the approximation quality.

\subsection*{Approximation Analysis}

In this subsection, we prove the approximation bound in \Cref{thm:mdsapprox}.
Specifically, we prove that the size of $\Malg$, the set of selected vertices, is within factor
$1 + 4 \cdot r \cdot f(r) \in O(r \cdot f(r))$ of the size of $\Mopt$,
a minimum distance-$r$ dominating set.

\begin{lemma}
  \label{thm:mdsapprox}
  If the graph class $\mathcal{C}$ has expansion $f(r)$ and girth at least $\mingirth$,
  then the set of vertices $\Malg$ selected by \Cref{alg:rmds} is small:
  ${\sizeof{\Malg}}/{\sizeof{\Mopt}} \in O(r \cdot f(r))$.
\end{lemma}
In the remainder of this subsection we prove \Cref{thm:mdsapprox}.
Note that this means that if $r$ is constant, then the approximation factor is constant, too.

We now analyze the behavior of \Cref{alg:rmds} on a particular graph $G \in \mathcal{C}$.
We begin by showing that the optimal solution implies a partition into
Voronoi cells~\cite{DBLP:books/sp/PreparataS85} which we will use for the rest.
First we define what a \emph{covering} vertex is.
Note that this can (and often is) different from the vertex selected by the algorithm.

\begin{definition}
  Let $c: V \to \Mopt$ be the mapping
  from vertices in $V$ to corresponding dominating vertices in $\Mopt$, breaking ties by distance (smaller), then by ID (smaller):
  \begin{equation}
    c(v) \defeq \argmin_{u \in N^{r}[v] \cap \Mopt} \{ (d(u, v), u) \} = \argmin_{u \in \Mopt} \{ (d(u, v), u) \}
  \end{equation}
\end{definition}

Again we order tuples lexicographically.
Now we can partition $V$ into Voronoi cells $H_{m} \defeq \{ v \in V | c(v) = m \}$ for each $m \in \Mopt$.

\begin{corollary}
  \label{lem:optpartitiongood}
  Each $H_{m}$ is connected and has radius at most $r$.
\end{corollary}
\begin{proof}
  As vertex $m$ dominates all vertices in $H_{m}$, we know that $H_{m}$ has radius at most $r$.
\end{proof}

Next, we use the high-girth property to show that the Voronoi cells behave nicely:

\begin{lemma}
    \label{lem:clusteristree}
    The subgraph induced by $H_{m}$ in $G$ is a tree.
\end{lemma}
\begin{proof}
    Assume towards contradiction that there is a cycle $C'$ in $H_{m}$.
    We will construct a cycle that has length at most $2r + 1$.

    Construct a BFS-tree of $H_{m}$ rooted in $m$.  Then the cycle $C'$ must contain an edge $e$ between $u, v \in H_{m}$.
    Consider the cycle that consists of the path from $u$ to $v$ along the BFS-tree, and the edge ${u, v}$.
    This cycle must have length at most $r + r + 1$, because the BFS-tree has depth at most $r$.
    This contradicts $G$ having a girth of at least $\mingirth$.
\end{proof}

\begin{lemma}
  \label{lem:clusterleqoneedge}
  For any two Voronoi cells $H_{m} \neq H_{n}$, there is at most one edge between them.
\end{lemma}
\begin{proof}
  Let $\{u,v\} \in E$ and $\{s,t\} \in E$ be two different edges between $H_{m}$ and $H_{n}$.
  W.l.o.g.~assume $c(u) = c(s) = m$ and $c(v) = c(t) = n$, and assume $v \neq t$ (but $u = s$ is possible).

  By \Cref{lem:optpartitiongood} we know that the subgraphs induced by $H_{m}$ and $H_{n}$ are each connected,
  so there must be a path $p_{m}$ entirely in $H_{m}$ between vertices $u$ and $s$, possibly of length 0.
  Likewise, a path $p_{n}$ must exist entirely in $H_{n}$ between vertices $v$ and $t$.
  The union of the paths and the assumed edges forms a cycle $C_{u,v,s,t}$,
  as no vertex can be repeated.
  We will now prove that $C_{u,v,s,t}$ is too small.

  The paths $p_{m}$ and $p_{n}$ have length at most $2r$, each.
  Therefore we have found the cycle $C_{u,v,s,t}$ to have length at most $4r+2$,
  in contradiction to the minimum girth $\mingirth$.
\end{proof}

Let $G' = (V', E')$ be the result of contracting $H_{m}$ to a single vertex, for each $m \in \Mopt$.

\begin{lemma}
  \label{lem:fewinterclusteredges}
  The edge set $E'$ is small: $\sizeof{E'} \leq f(r) \cdot \sizeof{M}$
\end{lemma}
\begin{proof}
  Using \Cref{lem:optpartitiongood}, we can apply the definition of the function $f(r)$:
  \begin{equation*}
    \sizeof{E'} = \frac{\sizeof{E'}}{\sizeof{V'}} \cdot \sizeof{V'} \leq f(r) \cdot \sizeof{M}
    \qedhere
  \end{equation*}
\end{proof}

Now we can take a closer look at the set of vertices $D$ actually selected by the algorithm.
We will construct two sets of bounded size such that their union covers $D$, and thereby bound the size of $D$.

\begin{definition}
  \label{def:didosets}
  We consider the set $D^{O}$ of vertices $v$ that were selected by a vertex $u$ in a different Voronoi cell (i.e.~$c(v) \neq c(u)$),
  and the possibly overlapping set $D^{I}$ of vertices $v$ that were selected by a vertex $u$ in the same Voronoi cell (i.e.~$c(v) = c(u)$):
  \begin{align*}
    D^{O} &\defeq \{ d \in \Malg \mid \exists v .\ v \text{ selected $d$ with $c(v) \neq c(d)$} \} \\
    D^{I} &\defeq \{ d \in \Malg \mid \exists v .\ v \text{ selected $d$ with $c(v) = c(d)$} \}
  \end{align*}
\end{definition}

Note that $D=D^{I} \cup D^{O}$, and $\sizeof{D} \leq \sizeof{D^{I}} + \sizeof{D^{O}}$. In order to prove \Cref{thm:mdsapprox}
it is therefore sufficient to show $D^{I}, D^{O} \in O(r \cdot f(r) \cdot \sizeof{M})$.
% Slack of factor two: If vertices are chosen both from the inside and from the outside,
% they will be counted twice.

First we will consider $D^{O}$, the set of vertices selected across Voronoi cells.
There are only few crossing edges, so there can only be few such selections:

\begin{lemma}
  The set $D^{O}$ of vertices selected across Voronoi cells is small: $D^{O} \in O(r \cdot f(r) \cdot \sizeof{\Mopt})$.
\end{lemma}
\begin{proof}
  \Cref{lem:fewinterclusteredges} and \Cref{lem:clusterleqoneedge} tell us
  that there are at most $f(r) \cdot \sizeof{\Mopt}$ edges across Voronoi cells.
  For each such edge at most $2r$ different vertices are announced (i.e.~$r$ in each direction),
  therefore there are at most $f(r) \sizeof{\Mopt} \cdot 2r$ many candidates for $D^{O}$.
\end{proof}

Next we prove a bound on the other set of \Cref{def:didosets}:
the set $D^{I}$ of vertices selected from within a Voronoi cell.
We will see that these always fall on the spanning tree inside Voronoi cells, which are small.
First we define the candidate set:

\begin{definition}
  For each Voronoi cell $H_{m}$, we define the set of vertices $\outtree_{m}$ as the union of
  all shortest paths $P_{u,m}$ between vertex $m$ and each vertex $u$ on the Voronoi boundary:%
  \begin{align*}
    \outtree_{m} &\defeq \bigcup_{\substack{\{u,v\} \in E \\ \mathclap{c(u)=m, c(v) \neq m}}} P_{u,m} &
    \allouttrees &\defeq \bigcup_{m \in M} \outtree_{m}
  \end{align*}
\end{definition}

This is well-defined due to \Cref{lem:clusteristree}.
Observe that $\outtree_{m}$ is not necessarily equal to $H_{m}$:
All leaves in $\outtree_{m}$ have edges in $G$ that lead outside the Voronoi cell.
If that is not true for a vertex in $H_{m}$, then it will not be in $\outtree_{m}$.

Next we want to prove in two steps that $D^{I}$ is contained in $\allouttrees$.

\begin{lemma}
  \label{lem:existsvertexbeyondreach}
  If $\sizeof{\Mopt} \geq 1$, then there is always a vertex beyond reach:
  \begin{equation}
    \forall v \in V\ \exists u \in V .\ d(u, v) = r + 1
  \end{equation}
  % It is tempting to prove that $u$ must be in a different Voronoi cell ($c(u) \neq c(v)$),
  % But this is not necessarily the case:  Let G be a path of length $2r+2$,
  % and let M be the two vertices in the middle.  Then in all cases,
  % the "far-away vertex" is in the same Voronoi cell.
\end{lemma}
\begin{proof}
  Assume towards contradiction there is a vertex $v$ for which no such vertex $u$ exists.
  Then there is also no vertex $u'$ with $d(u',v) > r+1$,
  because one could pick a shortest path, and construct such a $u$.
  Therefore $D' = \{v\}$ is a dominating set, in contradiction to the premise.
\end{proof}

With this we can show that only vertices in $\outtree_{m}$ are selected:

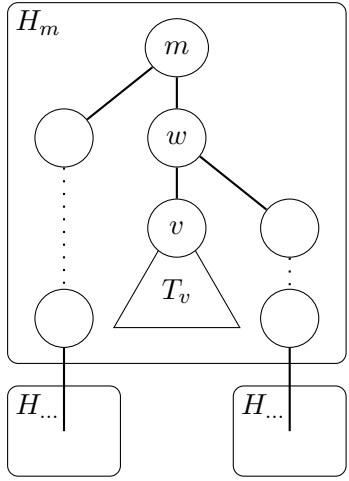
\begin{figure}
  \centering
  \begin{tikzpicture}[xscale=1.5,yscale=1.2]
    \tikzstyle{vertex}=[ellipse,draw,minimum width=22pt,minimum height=22pt]
    \tikzstyle{edge} = [draw,thick,-]

    % (3*xscale)x(4*yscale)
    \node[draw, minimum width=4.5cm, minimum height=4.8cm,rounded corners] (hm) at (0,-1.5) {};
    \node[inner sep=3pt,anchor=north west] (foobar) at (hm.north west) {$H_{m}$};
    % (1*xscale)x(1*yscale)
    \node[draw, minimum width=1.5cm, minimum height=1.2cm,rounded corners] (hl) at (-1,-4.25) {};
    \node[inner sep=3pt,anchor=north west] (foobar) at (hl.north west) {$H_{\ldots}$};
    % (1*xscale)x(1*yscale)
    \node[draw, minimum width=1.5cm, minimum height=1.2cm,rounded corners] (hr) at ( 1,-4.25) {};
    \node[inner sep=3pt,anchor=north west] (foobar) at (hr.north west) {$H_{\ldots}$};

    \node[regular polygon,regular polygon sides=3,draw,minimum height=55pt] at ( 0,-2.7) {$T_{v}$};
    \node[vertex] (m)   at ( 0, 0) {$m$};
    \node[vertex] (w)   at ( 0,-1) {$w$};
    \node[vertex,fill=white] (v)   at ( 0,-2) {$v$};
    \node[vertex] (alt) at (-1,-1) {};
    \node[vertex] (art) at ( 1,-2) {};
    \node[vertex] (alb) at (-1,-3) {};
    \node[vertex] (arb) at ( 1,-3) {};
    \path[edge] (alt) -- (m) -- (w) -- (art);
    \path[edge] (w) -- (v);
    \path[edge, loosely dotted] (alt) -- (alb);
    \path[edge, loosely dotted] (art) -- (arb);

    \path[edge] (alb) -- (hl.center);
    \path[edge] (arb) -- (hr.center);
\end{tikzpicture}
  \caption{Typical vertex layout in proof of \Cref{lem:selectouttree}.
    The identity of vertex $u$ does not matter, hence it is not shown}
  \label{fig:diisouttree}
\end{figure}

\begin{lemma}
  \label{lem:selectouttree}
  Let $v$ be a vertex selected by $u$ in the same Voronoi cell.
  Then $v \in \outtree_{c(v)}$:%
  \begin{equation}
    \big( \varid{u} = v \land c(u) = c(v) \big) \implies v \in \outtree_{c(v)}
  \end{equation}
\end{lemma}
\begin{proof}
  Assume towards contradiction that $v \notin \outtree_{c(v)}$.  For brevity, let $m \defeq c(v)$.
  Observe that $m \neq v$, because $m \in \outtree_{m}$.
  Let $w$ be the next vertex on a shortest path from $v$ towards $m$; possibly $m$ itself.
  We will now analyze the properties of vertex $w$ and conclude that vertex $u$ should not have selected $v$.
  Refer to \Cref{fig:diisouttree} for an overview.

  By \Cref{lem:clusteristree}, the subgraph induced by $H_{m}$ is a tree.
  If we root this $H_{m}$-tree at vertex $m$, we can denote the subtree rooted at $v$ as $T_{v}$.
  This subtree has depth at most $r - 1$, so $w$ covers the entire subtree:
  $T_{v} \subseteq N^{r}(w)$.  All vertices $x \in V \setminus T_{v}$ are closer to $w$ than to $v$,
  as all paths from $x$ to $v$ must go through $w$.
  So the neighborhood of $v$ is included in the neighborhood of $w$:
  $N^{r}[v] \subseteq N^{r}[w]$.

  Now we can use \Cref{lem:existsvertexbeyondreach}:
  % Note that N^{r}() may partially lie outside of H_{m}.
  There must be a vertex $t$ that has distance $r+1$ to vertex $v$, so $t \notin N^{r}(v)$.
  This means that $t \notin T_{v}$, and therefore $t \in N^{r}(w)$.
  In summary, the degree of $w$ is strictly larger: $\sizeof{N^{r}(w)} > \sizeof{N^{r}(v)}$,
  meaning that vertex $u$ would prefer selecting $w$ over $v$.

  All that is left is to show that vertex $u$ is indeed able to select vertex $w$:
  If $u \in T_{v}$, then the maximum depth of $r-1$ means the distance to $w$ is at most $r$.
  If $u \notin T_{v}$, then $w$ is on every shortest path between $u$ and $v$,
  and therefore in reach, too.

  This leads to a contradiction: Vertex $u$ selected $v$,
  although vertex $w$ is in reach, and has a strictly larger distance-$r$ neighborhood,
  and should be preferred by the algorithm.
\end{proof}

\begin{corollary}
  \label{cor:diisallouttrees}
  Selections within a Voronoi cell are restricted to $\allouttrees$: $D^{I} \subseteq \allouttrees$
\end{corollary}

Now we can show a bound on $D^{I}$ by proving it on $\allouttrees$:

\begin{lemma}
  \label{lem:outtreesaresmall}
  The set $\allouttrees$ is small: $\sizeof{\allouttrees} \leq (1 + 2 r \cdot f(r)) \sizeof{M}$
\end{lemma}
\begin{proof}
  Consider an arbitrary but fixed $\{v, u\} \in E$ with $c(v) \neq c(u)$.
  Each path $P_{v,c(v)}$ has at most $r$ vertices not in $M$,
  because it is a shortest path, and by construction all vertices are dominated by $c(v)$.
  Each edge in $E'$ corresponds to at most two such paths.
  With \Cref{lem:fewinterclusteredges},
  this bounds the number of paths to at most $2 f(r) \sizeof{M}$.

  Therefore $\allouttrees$ contains at most $2 r \cdot f(r) \sizeof{M} + \sizeof{M}$ vertices.
\end{proof}

\begin{corollary}
  The set $D^{I}$ is small: $\sizeof{D^{I}} \leq (1 + 2 r \cdot f(r)) \sizeof{M}$
\end{corollary}
\begin{proof}
  Follows from \Cref{cor:diisallouttrees} and \Cref{lem:outtreesaresmall}.
\end{proof}

As both $\sizeof{D^{I}}$ and $\sizeof{D^{O}}$ are in $O(r \cdot f(r) \cdot \sizeof{M})$,
this concludes the proof of \Cref{thm:mdsapprox}, and thus \Cref{thm:algorithm}.
More specifically, we proved the upper bound $\left( 1 + 4 \cdot r \cdot f(r) \right) \cdot \sizeof{M}$ on $\sizeof{D}$.

\subsection*{Tightness of Approximation}

The previous subsection proved that the algorithm is a $O(r \cdot f(r))$ approximation.
Is it possible that the algorithm actually performs significantly better what the analysis guarantees?
This subsection proves that there are graphs for which the algorithm yields a $\Omega(r \cdot f(r))$ approximation,
meaning that the above analysis of the algorithm is asymptotically tight.

\begin{lemma}
  \label{lem:rmdsalgoistight}
  There is a computable function $f\colon \N\rightarrow \N$
  and a class of graphs $\mathcal{C}$ of expansion $f(r)$ and girth at
  least $\mingirth$ such that \Cref{alg:rmds} takes $\Omega(r)$ rounds
  and provides an $\Omega(r\cdot f(r))$-approximation of minimum
  distance-$r$ dominating set on $\mathcal{C}$. (Cf.\ \Cref{thm:algorithm}.)
\end{lemma}
\begin{proof}
  The rest of this subsection constructs such a $\mathcal{C}$ consisting of graphs $G_{r,f(r)}$
  for all values of $r \geq 1$ and $f(r) \geq 2$.
\end{proof}

This does not mean that the problem is hard.
It only shows that in the worst case, the presented algorithm may
use up the approximation slack.

The construction is a modified version of the subdivided biclique.
Let $X$ and $Y$ be two disjoint sets of vertices, each of size $2 f(r)$.
For each pair $(x, y) \in X \times Y$, connect them with a path $P_{x,y}$ of $2r$ vertices, such that $d(x, y) = 2r + 1$.
This means that no vertex can simultaneously cover $x$ and $y$, i.e., is within distance $r$ of both $x$ and $y$.
For each $x \in X, y \in Y$, create a set $B_{x,y}$ of $k = 2r \cdot f(r)$ new vertices,
and connect each vertex in $B_{x,y}$ by a single edge to the vertex closest to $x$ of each path.
Let $V$ be the union of all these sets, and $E$ as described, then $G_{r,f(r)}=(V,E)$ is the constructed graph.

First we prove that the graph class satisfies all requirements.

\begin{lemma}
  For arbitrary but fixed values of $r \geq 1$ and $f(r) \geq 2$,
  the graph class $\mathcal{C}$ has expansion $f(r)$ and girth at
  least $\mingirth$.
\end{lemma}
\begin{proof}
Let $G_{r,f(r)} \in \mathcal{C}$ be a fixed graph from the constructed graph class.
The girth is at least $4 \cdot (2r + 1) > \mingirth$,
as a cycle needs to pass through at least two vertices from $X$ and two vertices from $Y$.
The low expansion can be shown by contracting as much as possible around
all vertices in $X \cup Y$, which results in the biclique $K_{2f(r),2f(r)}$,
with $4f(r)$ vertices and $4f(r)^{2}$ edges.
Therefore, the constructed graph has $f^{G}(r) \geq f(r)$.
As this is the optimal contraction choice, this also shows $f^{G}(r) = f(r)$.
\end{proof}

Next we show that this graph class causes worst-case behavior.
The running time is trivial:

\begin{lemma}
  \Cref{alg:rmds} runs in $\Omega(r)$ rounds on graphs in $\mathcal{C}$.
\end{lemma}
\begin{proof}
  Follows from the construction of \Cref{alg:rmds}.
\end{proof}

Next, we show that the algorithm computes a large dominating set, compared to the optimum:

\begin{lemma}
  The proposed \Cref{alg:rmds}
  provides an $\Omega(r\cdot f(r))$-approximation of minimum
  distance-$r$ dominating set on $\mathcal{C}$.
\end{lemma}
\begin{proof}
By construction, $X \cup Y$ is a dominating set, meaning $\sizeof{\Mopt} \leq 4 f(r)$.
Therefore, it suffices to show that $\sizeof{\Malg} \geq 4 r \cdot f(r)^{2}$.

We do so by simulating the algorithm on $G$.
We only need to consider the vertices selected by vertices on the paths do.
Specifically, pick a specific path $P_{x,y}$ between $x \in X$ and $y \in Y$.
Vertices $v_{x}$ closer to $x$ than to $y$ cover the attached vertices $B_{x,y}$,
so $\sizeof{N^{r}(v_{x})} \geq 2r + k = 2r + 2r \cdot f(r)$.
The vertices closer to vertex $x$ cover more of the other paths ending in $x$,
each step increases $\sizeof{N^{r}(v_{x})}$ by at least $2f(r)-1$,
and loses at most 1 vertex out of sight in the $y$ direction.
Note that we ignore the vertices in $B_{x,y'}$ with $y' \neq y$, which would only make this argument stronger.
The important property is that $\sizeof{N^{r}(v_{x})}$ strictly increases towards $x$, among vertices $v_{x}$ with $d(v_{x},x)<d(v_{x},y)$.

Each vertex $v_{y}$ closer to $y$ than to $x$ does not cover the attached vertices $B_{x,y}$ close to vertex $x$,
as distance $r$ from them would imply distance $r$ to $x$.
We can compute $\sizeof{N^{r}(v_{y})} \leq r + N^{r}(v_{r}) + 1 - 1 = r + r \cdot 2 f(r) < 2r + 2r \cdot f(r) \leq \sizeof{N^{r}(v_{x})}$,
so vertex $v_{y}$ will choose some vertex $v_{x}$.
As we already established, $\sizeof{N^{r}(v_{x})}$ increases with decreasing distance to $x$.
Therefore, each $v_{y}$ will select the vertex closest to $x$,
meaning at least half of each path will be selected, specifically the one on the $v_{l}$ side.

In total this means the algorithm selects at least $r$ vertices per path, and there is one such path for each $X \times Y$ combination.
Hence $\sizeof{\Malg} \geq r \cdot 4 \cdot f(r)^{2}$.  Recall that $\sizeof{\Mopt} \leq 4 f(r)$,
so the algorithm achieves an approximation factor of at least $r \cdot f(r)$ for the constructed graph.
Compared with the upper bound
of $1 + (4 r \cdot f(r))$ this is asymptotically tight.
\end{proof}

This concludes the proof of \Cref{lem:rmdsalgoistight} (tightness of approximation).

\section{Lower Bound}\label{sec:hardness}

In this section, we will prove that significantly better approximation of the
problem is hard.
Intuitively speaking, this is because symmetry cannot be broken in $o(\log^{\ast} n)$ rounds,
and without that it is hard to construct any non-trivial \distr.
Hence, this section is dedicated to the proof of~\Cref{thm:lowerbound}:

\begin{apthm}{thm:lowerbound}
    % Also change the copy!
    Assume an arbitrary but fixed $\delta > 0$ and $r > 1$, with $r \in o(\log^{\ast} n)$.
    Then there is no deterministic \local algorithm that finds in $O(r)$ rounds
    a $(2r + 1 - \delta)$-approximation of \distr for all $G \in \mathcal{C}$,
    where $\mathcal{C}$ is the class of cycles of length $\gg \mingirth$.
\end{apthm}

As we will see later, the trivial \distr $V$ (i.e., the set of all vertices),
is a $(2r + 1)$-approximation in the case of cycles.

This has been proven implicitly in the work of~\cite{ds-planar}.
However, we find it simpler to provide a new proof
tailored for our setting, but only for $n$ being an multiple of $2r+1$.
In essence, we will show a reduction from \enquote{large} independent
set to \distr, on the graph class of cycles.
Intuitively speaking, any algorithm that does significantly better than the trivial dominating set \emph{anywhere} on the cycle leads to a linear sized independent set;
and the bound is constructed such that the algorithm needs to do
better than trivial somewhere indeed. 

The idea is simple: find a
\distr $D$ on cycle $C$; we know two consecutive vertices of $D$ on $C$
are of distance at most $2r+1$ from each other, hence, these vertices
help us to break the symmetry and as $r\in o(\log^*n)$ it yields
to an independent set of size $O(n)$ in $o(\log^*n)$ rounds. In the remainder we
formalize this argument.

Assume towards contradiction that \loweralg is such a deterministic distributed algorithm,
which finds a \distr in $G\in\mathcal{C}$ of size at most $(2r + 1 - \delta)\sizeof{\Mopt}$,
where $\Mopt$ is a minimum \distr.

We show that \loweralg can be used to construct an algorithm violating known lowerbounds on
\enquote{large} independent set~\cite{ds-planar,10.1007/978-3-540-87779-0_6}:

\begin{lemma}[Lemma $4$ of~\cite{10.1007/978-3-540-87779-0_6}]
    \label{lem:lisishard}
    There is no deterministic distributed algorithm that finds an independent set
    of size $\Omega(n / \log^{\ast} n)$ in a cycle on $n$ vertices in $o(\log^{\ast} n)$ rounds.
\end{lemma}

We present the reduction algorithm in \Cref{alg:mistormdsreduction}.

\begin{algorithm}[h]
\caption{\congest computation of an IS on a cycle $G \in \mathcal{C}_{r}$, for each $v$ in parallel}
\begin{algorithmic}[1]
    \STATE Compute a \distr $D$ by simulating \loweralg.
    \STATE Determine the connected components $V \setminus D$.
    \FOR{each component $C_i$}
        \STATE Determine the two adjacent vertices to $C_i$,
        i.e.~$u,v\in N(C_i)$.
        \STATE Let $u$ be the vertex with the lower ID, name it
        representor of $C_i$.
        \STATE All vertices of odd distance to $u$ in $C_i$ join $I$.
    \ENDFOR{}
    \RETURN $I$
\end{algorithmic}\label{alg:mistormdsreduction}
\end{algorithm}

We begin by showing basic correctness:

\begin{lemma}
    \Cref{alg:mistormdsreduction} runs in $o(\log^{\ast} n)$ rounds.
\end{lemma}
\begin{proof}
    By assumption, \loweralg executes in $O(r)$ rounds.
On the other hand, observe that each vertex in $D$ only covers up to a distance of $r$.
    Because $D$ is a dominating set, all component must have length at
    most $2r$.
    Hence, discovering the adjacent vertex
    of lowest ID can be done in $O(r)$, as well as propagating the distance information.
    By construction $r \in o(\log^{\ast} n)$,
    so~\Cref{alg:mistormdsreduction} takes $o(\log^{\ast} n)$ rounds.
\end{proof}

\begin{lemma}
    \Cref{alg:mistormdsreduction} computes set $I$ which is an independent set.
\end{lemma}
\begin{proof}
For two distinct vertices $u,v\in I$, if they belong to different
components, then there is no edge between them, otherwise if they are
in the same component, their distance is at least $2$ as they are distinct
vertices of odd distance from their representor. Hence, $I$ is an independent set.
\end{proof}

Now we can show that this yields a large independent set:

\begin{lemma}
    \label{lem:lowermalgissmall}
    The dominating set is not too large: $\sizeof{D} \leq (1 - \delta') n$ for some $\delta' > 0$.
\end{lemma}
\begin{proof}
    By assumption, we know $\sizeof{D} \leq (2r + 1 - \delta) \sizeof{M}$, where $M$ is the minimum \distr.
    Construct $M'$ by picking every $2r+1$-th vertex, so that $\sizeof{M'} = n / (2r+1)$.
    Note that $M'$ is a \distr, so we have $\sizeof{M} \leq \sizeof{M'}$.
    Together we get $\sizeof{D} \leq (2r + 1 - \delta) n / (2r + 1) = (1 - \delta') n$, for $\delta' \defeq 1 / (2r + 1) > 0$.
\end{proof}

\begin{lemma}
    \label{lem:lowerlargeindepset}
    The set $I$ is large: $\sizeof{I} \in \Omega(n / \log^{\ast} n)$
\end{lemma}
\begin{proof}
    Many vertices must be part of some component: $\sizeof{V \setminus D} \geq \delta' n$ for some $\delta' > 0$ by \Cref{lem:lowermalgissmall}.
    At least half of those vertices are taken into $I$, thus $\sizeof{I} \geq \delta' n / 2 \in \Omega(n / \log^{\ast} n)$.
\end{proof}

\begin{proof}[Proof of \Cref{thm:lowerbound}]
    Follows immediately from \Cref{lem:lowerlargeindepset}, as it contradicts \Cref{lem:lisishard}.
    Therefore, the algorithm \loweralg cannot exist.
\end{proof}

Note that this does not preclude randomized algorithms.
This is because randomized algorithms can indeed achieve a better approximation quality,
at least on cycles, by randomly joining the dominating set with sufficiently small
probability if necessary, for several rounds, and finally all uncovered vertices join.

\section{Conclusion}\label{sec:conclude}

In this paper, we have analyzed the problem of \distr
for $r \in o(\log^{\ast} n)$
in a deterministic setting on graphs of  bounded expansion $f(r)$
(e.g.~on planar graphs$f(r)=3$ for every $r$) and high girth.
We provided a simple \congest algorithm and proved that it achieves an
approximation factor
of $O(r \cdot f(r))$ and a running time of $O(r)$.
In particular, for constant $r$, it provides a constant
factor approximation in a constant number of rounds.

For the lower bound, we have shown that the standard $o(log^\ast n)$ lower
bound on bounded degree graphs for other problems, also holds here,
even if $r$ is a non-constant. 
This means that one cannot do much better than~\Cref{alg:rmds} up to a factor of $f(r)$.
For $r \in \Omega(\log^{\ast} n)$, the algorithm still works, but the
lower bound no longer applies.

For the upper bound, another interesting and useful problem is the
distance-$r$ independent set problem on sparse graphs. Especially
because of the tight relation between the independent set and coloring
and consequently $r$-hop network decomposition, it is valuable to
further research these topics on sparse graphs.

For the lower bound, we showed that $r$ plays a role in any
algorithm, but we do not know whether the existence of the expansion
function of the graph class is essential. By the lower bound of Kuhn
et al.~\cite{Kuhn:2016:LCL:2906142.2742012} dependency on sparsity
is clear, but to what extend? I.e.\ is it possible to provide a constant
factor approximation in constant number of rounds for the \distr
problem on graphs of bounded arboricity with high girth?

\subparagraph*{Acknowledgments: } We are grateful to thank our friend Christoph
Lenzen and anonymous reviewers for valuable feedback and improvements.

\bibliographystyle{alpha}
\bibliography{references}

\end{document}